\documentclass{llncs}
\usepackage{url}
\usepackage{amsmath}
\usepackage{amsfonts}
\usepackage{array,multirow}
\usepackage{enumerate}
\usepackage{graphicx}
\usepackage{fixltx2e}
\usepackage{tikz}
\usetikzlibrary{positioning}
\usepackage{latexsym}
\usepackage[T1]{fontenc}
\usepackage{hyperref}

\def\MPD{\mathrm{MPD}}
\def\MFP{\mathrm{M5P}}
\def\FDM{\mathrm{4DM}}

\let\realendproof=\endproof
\def\endproof{\hspace*{\fill}$\Box$\realendproof}

\definecolor{xxxcolor}{rgb}{0.8,0,0}

\newcounter{section-preserve}
\newcounter{theorem-preserve}
\newcommand{\blank}[1]{}
\newtoks\magicAppendix
\magicAppendix={}
\newtoks\magictoks
\newif\iflater
\laterfalse
\long\def\later#1{\magictoks={#1}%
  \edef\magictodo{\noexpand\magicAppendix={\the\magicAppendix \par
    \the\magictoks%
  }}
  \magictodo}
\long\def\both#1{\magictoks={#1}%
  \edef\magictodo{\noexpand\magicAppendix={\the\magicAppendix \par
    \noexpand\setcounter{theorem-preserve}{\noexpand\arabic{theorem}}%
    \noexpand\setcounter{theorem}{\arabic{theorem}}%
    \noexpand\setcounter{section-preserve}{\noexpand\arabic{section}}%
    \noexpand\setcounter{section}{\arabic{section}}%
    \noexpand\let\noexpand\oldsection=\noexpand\thesection
    \noexpand\def\noexpand\thesection{\thesection}
    \noexpand\let\noexpand\oldlabel=\noexpand\label
    \noexpand\let\noexpand\label=\noexpand\blank
    \the\magictoks%
    \noexpand\setcounter{theorem}{\noexpand\arabic{theorem-preserve}}%
    \noexpand\setcounter{section}{\noexpand\arabic{section-preserve}}%
    \noexpand\let\noexpand\thesection=\noexpand\oldsection
    \noexpand\let\noexpand\label=\noexpand\oldlabel
  }}
  \magictodo
  \the\magictoks}
\def\magicappendix{\latertrue \the\magicAppendix}

\newif\ifabstract
\abstractfalse
\newif\iffull
\ifabstract \fullfalse \else \fulltrue \fi

\iffull
 \long\def\both#1{#1}
 \let\later=\both
 \def\magicappendix{}
\fi

\title{Dissection with the Fewest Pieces is Hard, \\ Even to Approximate}
\author{
Jeffrey Bosboom\inst{1} \and 
Erik D. Demaine\inst{1} \and 
Martin L. Demaine\inst{1} \and
Jayson Lynch\inst{1} \and
Pasin Manurangsi\inst{2}\thanks{Part of this work was completed while the author was at Massachusetts Institute of Technology and Dropbox, Inc.} \and
Mikhail Rudoy\inst{1} \and
Anak Yodpinyanee\inst{1}\thanks{Research supported by NSF grants CCF-1217423, CCF-1065125, and CCF-1420692.}}
\institute{
        Computer Science and AI Laboratory,
	Massachusetts Institute of Technology \\
	32 Vassar St., Cambridge, MA 02139, USA \\
        \url{{jbosboom,edemaine,mdemaine,jaysonl,mrudoy,anak}@mit.edu}
	\and
	University of California, %Berkeley,
	Berkeley, CA 94720, USA \\
        \url{pasin@berkeley.edu}}

\date{\today}

\begin{document}
\maketitle

\begin{abstract}
  We prove that it is NP-hard to dissect one simple orthogonal polygon into
  another using a given number of pieces, as is approximating the fewest
  pieces to within a factor of $1+1/1080-\varepsilon$.
\end{abstract}

\section{Introduction}

We have known for centuries how to dissect any polygon $P$ into any
other polygon $Q$ of equal area, that is, how to cut $P$ into
finitely many pieces and re-arrange the pieces to form~$Q$
\cite{Frederickson-1997,Lowry-1814,Wallace-1831,Bolyai-1832,Gerwien-1833}.
But we know relatively little about how many pieces are necessary.
For example, it is unknown whether a square can be dissected into an
equilateral triangle using fewer than four pieces
\cite[pp.~8--10]{Dudeney-1902-hinged,Frederickson-2002}.
Only recently was it established that a pseudopolynomial number of
pieces suffices~\cite{FeatureSize_EGC2011f}.

In this paper, we prove that it is NP-hard even to approximate the minimum
number of pieces required for a dissection, to within some constant ratio.
While perhaps unsurprising, this result is the first analysis
of the complexity of dissection.  We prove NP-hardness even when
the polygons are restricted to be simple (hole-free) and orthogonal.
The reduction holds for all cuts that leave the resulting pieces connected,
even when rotation and reflection are permitted or forbidden.

Our proof significantly strengthens the observation (originally made by the
Demaines during JCDCG'98) that the second half of dissection---re-arranging
given pieces into a target shape---is NP-hard: the special case of exact
packing rectangles into rectangles can directly simulate \textsc{3-Partition}
\cite{Jigsaw_GC}.
Effectively, the challenge in our proof is to construct a polygon for which
any $k$-piece dissection must cut the polygon at locations we desire, so that
we are left with a rectangle packing problem.

\section{The Problems} \label{sec-prob}

\subsection{Dissection}

We begin by formally defining the problems involved in our proofs, starting with \textsc{$k$-Piece Dissection}, which is the central focus of our paper.

\begin{definition}
	\textsc{$k$-Piece Dissection} is the following decision problem:

	\textsc{Input:} two polygons $P$ and $Q$ of equal area, and a positive integer $k$.

	\textsc{Output:} whether $P$ can be cut into $k$ pieces such that these $k$ pieces can be packed into $Q$ (via translation, optional rotation, and optional reflection).
\end{definition}

To prevent ill-behaved cuts, we require every piece to be a \emph{Jordan region (with holes)}: the set of points interior to a Jordan curve $e$ and exterior to $k \geq 0$ Jordan curves $h_1, h_2, \dots, h_k$, such that $e, h_1, h_2, \dots, h_k$ do not meet.
There are two properties of Jordan regions that we use in our proofs. First, Jordan regions are Lebesgue measurable; we will refer to the Lebesgue measure of each piece as its area. Second, a Jordan region is path-connected. For brevity, we refer to path-connected as connected throughout the paper.

Next we define the optimization version of the problem, \textsc{Min Piece Dissection}, in which the objective is to minimize the number of pieces.

\begin{definition}
	\textsc{Min Piece Dissection} is the following optimization problem:

	\textsc{Input:} two polygons $P$ and $Q$ of equal area.

	\textsc{Output:} the smallest positive integer $k$ such that $P$ can be cut into $k$ pieces such that these $k$ pieces can be packed into $Q$.
\end{definition}

\subsection{5-Partition}

Our NP-hardness reduction for \textsc{$k$-Piece Dissection} is from
\textsc{5-Partition}, a close relative of \textsc{3-Partition}.

\begin{definition}
	\textsc{5-Partition} is the following decision problem:

	\textsc{Input:} a multiset $A = \{a_1, \dots, a_n\}$ of $n = 5m$ integers.

	\textsc{Output:} whether $A$ can be partitioned into $A_1, \dots, A_m$ such that, for each $i = 1, \dots, m$, $\sum_{a \in A_i} a = p$ where $p = \left(\sum_{a \in A} a \right)/m$.
\end{definition}

Throughout the paper, we assume that the partition sum $p$ is an integer; otherwise, the instance is obviously a \textsc{No} instance.

Garey and Johnson~\cite{Garey-Johnson-1975} originally proved NP-completeness of \textsc{3-Partition}, a problem similar to \textsc{5-Partition} except that 5 is replaced by 3. In their book~\cite{Garey-Johnson-1979}, they show that \textsc{4-Partition} is NP-hard; this result was, in fact, an intermediate step toward showing that \textsc{3-Partition} is NP-hard. It is easy to reduce \textsc{4-Partition} to \textsc{5-Partition} and thus show it also NP-hard.\footnote{Given a \textsc{4-Partition} instance $A = \{a_1, \dots, a_n\}$, we can create a \textsc{5-Partition} instance by setting $A' = \{na_1, \dots, na_n, 1, \dots, 1\}$ where the number of $1$s is $n/4$.}

Our reduction would work from 3-Partition just as well as 5-Partition. The advantage of the latter is that we can analyze the following optimization version.

\begin{definition}
	\textsc{Max 5-Partition} is the following optimization problem:

	\textsc{Input:} a multiset $A = \{a_1, \dots, a_n\}$ of $n = 5m$ integers.

	\textsc{Output:} the maximum integer $m'$ such that there exist disjoint subsets $A_1, \dots, A_{m'}$ of $A$ such that, for each $i = 1, \dots, m'$, $\sum_{a \in A_i} a = p$ where $p = \frac{5}{n}\left(\sum_{a \in A} a \right)$.
\end{definition}

\subsection{4-Dimensional Matching}

While \textsc{5-Partition} is known to be NP-hard, we are not aware of any results on the hardness of approximating \textsc{Max 5-Partition}.  We derive the result ourselves by reducing from \textsc{4-Uniform 4-Dimensional Matching} (\textsc{4DM}).

\begin{definition}
	\textsc{4-Uniform 4-Dimensional Matching} (or \textsc{4DM}) is the following optimization problem:

	\textsc{Input:} a 4-uniform 4-partite balanced hypergraph $H = (V^1, V^2, V^3, V^4, E)$.

	\textsc{Output:} a matching of maximum size in $H$, i.e., a set $E' \subseteq E$ of maximum size such that none of the edges in $E'$ share the same vertex.
\end{definition}

Hazan, Safra, and Schwartz~\cite{Hazan-Safra-Schwartz-2003} proved inapproximability of \textsc{4-Uniform 4-Dimensional Matching}.  We use this to prove hardness of approximating \textsc{Max 5-Partition}, which we then reduce to \textsc{Min Piece Dissection} to establish the hardness of its approximation.

\subsection{Gap Problems}

We show that our reductions have a property stronger than approximation preservation called \emph{gap preservation}.
Let us define the gap problem for an optimization problem, a notion widely used in hardness of approximation.

\begin{definition}
	For an optimization problem $P$ and parameters $\beta > \gamma$ (which may be functions of $n$), the \textsc{Gap}$_P[\beta, \gamma]$ problem is to distinguish whether the optimum of a given instance of $P$ is at least $\beta$ or at most $\gamma$.  The input instance is guaranteed to not have an optimum between $\beta$ and $\gamma$.
\end{definition}

If \textsc{Gap}$_P[\beta, \gamma]$ is NP-hard, then it immediately follows that approximating $P$ to within a factor of $\beta/\gamma$ of the optimum is also NP-hard.  This result makes gap problems useful for proving hardness of approximation.

\section{Main Results}

Now that we have defined the problems, we state our main results.

\begin{theorem} \label{thm-main}
	\textsc{$k$-Piece Dissection} is NP-hard.
\end{theorem}

We do not know whether \textsc{$k$-Piece Dissection} is in NP (and thus is NP-complete).  We discuss the difficulty of showing containment in NP in Section~\ref{sec-open}.

We also prove that the optimization version, \textsc{Min Piece Dissection}, is hard to approximate to within some constant ratio:

\begin{theorem} \label{thm-main-approx}
	There is a constant $\varepsilon_{\MPD} > 0$ such that it is NP-hard to approximate \textsc{Min Piece Dissection} to within a factor of $1 + \varepsilon_{\MPD}$ of optimal.\footnote{The best $\varepsilon_{\MPD}$ we can achieve is $1/1080 - \varepsilon$ for any $\varepsilon \in (0, 1/1080)$.}
\end{theorem}

Both results are based on essentially the same reduction, from \textsc{5-Partition} for Theorem~\ref{thm-main} or from \textsc{Max 5-Partition} for Theorem~\ref{thm-main-approx}. We present the common reduction in Section~\ref{sec-reduction}. We then prove Theorem~\ref{thm-main} and Theorem~\ref{thm-main-approx} in Sections~\ref{sec-proof} and~\ref{sec-proof-approx} respectively.

Restricting the kinds of polygons given as input, the kinds of cuts allowed, and the ways the pieces can be packed gives rise to many variant problems.  Section~\ref{sec-variant} explains for which variants our results continue to hold.

\section{The Reduction} \label{sec-reduction}

This section describes a polynomial-time reduction from \textsc{5-Partition} to \textsc{$k$-Piece Dissection} and states a lemma crucial to both of our main proofs later in the paper.  The proof of the lemma is deferred to \ifabstract Appendix\else Section\fi~\ref{sec-lem-proof}.

\paragraph{Reduction from \textsc{5-Partition} to \textsc{$k$-Piece Dissection}.} Let $A = \{a_1, \dots, a_n\}$ be the given \textsc{5-Partition} instance and let $p = \frac{5}{n}\Sigma_{a \in A} a$ denote the target sum. Let $d_s = 12 (\max_{a \in A} a + p)$ and $d_t = (n-1)d_s + \Sigma_{a \in A} a + 2\max_{a \in A} a$.  We create a source polygon $P$ consisting of \emph{element rectangles} of width $a_i$ and height 1 for each $a_i \in A$ spaced $d_s$ apart, connected below by a rectangular \emph{bar} of width $\Sigma_{a \in A} a + (\frac{n}{5} - 1)d_t$ and height $\delta = \frac{1}{10\Sigma_{a \in A} a + 2(\frac{n}{5} - 1)d_t}$.  The first element rectangle's left edge is flush with the left edge of the bar; the bar extends beyond the last element rectangle. Our target polygon $Q$ consists of $\frac{n}{5}$ \emph{partition rectangles} of width $p$ and height 1 spaced $d_t$ apart, connected by a bar of the same dimensions as the source polygon's bar.  The first partition rectangle's left edge and last partition rectangle's right edge are flush with the ends of the bar.  Both polygons' bars have the same area and the total area of the element rectangles equals the total area of the partition rectangles, so the polygons have the same area. Finally, let the number of pieces $k$ be $n$.

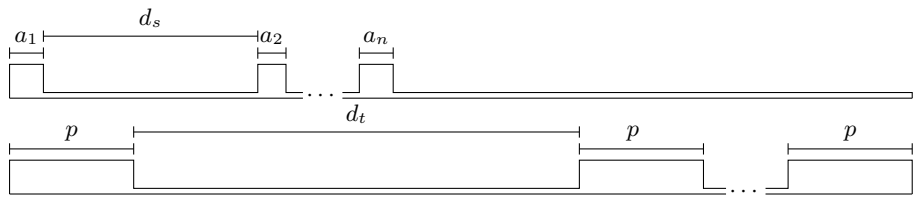
\begin{figure}[h]
\centering
\begin{tikzpicture}[scale=0.75]
% source polygon
% upper edge
\draw (1,0.9)--(1,1.5)--(1.6,1.5)--(1.6,1)--(5.4,1)--(5.4,1.5)--(5.9,1.5)--(5.9,1)--(6.2,1);
\draw (6.9,1)--(7.2,1)--(7.2,1.5)--(7.8,1.5)--(7.8,1)--(17,1)--(17,0.9);
% lower edge
\draw (1,0.9)--(6.2,0.9);
\draw (6.9,0.9)--(17,0.9);
% measurement line
\draw (1,1.7)--(1.6,1.7);
\draw (1,1.6)--(1,1.8);
\draw (1.6,1.6)--(1.6,1.8);
\draw (5.4,1.7)--(5.9,1.7);
\draw (5.4,1.6)--(5.4,1.8);
\draw (5.9,1.6)--(5.9,1.8);
\draw (7.2,1.7)--(7.8,1.7);
\draw (7.2,1.6)--(7.2,1.8);
\draw (7.8,1.6)--(7.8,1.8);
\draw (1.6,2)--(5.4,2);
\draw (1.6,1.9)--(1.6,2.1);
\draw (5.4,1.9)--(5.4,2.1);
% text
\node at (6.55,0.95) {$\cdots$};
\node[above] at (1.3,1.7) {$a_1$};
\node[above] at (5.65,1.7) {$a_2$};
\node[above] at (7.5,1.7) {$a_n$};
\node[above] at (3.5,2) {$d_s$};

% target polygon
% upper edge
\draw (1,-0.8)--(1,-0.2)--(3.2,-0.2)--(3.2,-0.7)--(11.1,-0.7)--(11.1,-0.2)--(13.3,-0.2)--(13.3,-0.7)--(13.7,-0.7);
\draw (14.4,-0.7)--(14.8,-0.7)--(14.8,-0.2)--(17,-0.2)--(17,-0.8);
% lower edge
\draw (1,-0.8)--(13.7,-0.8);
\draw (14.4,-0.8)--(17,-0.8);
% measurement lines
\draw (1,0)--(3.2,0);
\draw (1,-0.1)--(1,0.1);
\draw (3.2,-0.1)--(3.2,0.1);
\draw (11.1,0)--(13.3,0);
\draw (11.1,-0.1)--(11.1,0.1);
\draw (13.3,-0.1)--(13.3,0.1);
\draw (14.8,0)--(17,0);
\draw (14.8,-0.1)--(14.8,0.1);
\draw (17,-0.1)--(17,0.1);
\draw (3.2,0.3)--(11.1,0.3);
\draw (3.2,0.2)--(3.2,0.4);
\draw (11.1,0.2)--(11.1,0.4);
% text
\node at (14.05,-0.75) {$\cdots$};
\node[above] at (2.1,0) {$p$};
\node[above] at (12.05,0) {$p$};
\node[above] at (15.9,0) {$p$};
\node[above] at (7.15,0.3) {$d_t$};
\end{tikzpicture}
\caption{The source polygon $P$ (above) and the target polygon $Q$ (below) are shown (not to scale). Length $d_t$ is longer than the distance between the leftmost edge of the leftmost element rectangle and the rightmost edge of the rightmost element rectangle.}
\end{figure}

\paragraph{Reduction from \textsc{Max 5-Partition} to \textsc{Min Piece Dissection}.} The optimization problem uses the same reduction as the decision problem, except that we do not specify $k$ for the optimization problem.

The idea behind our reduction is to force any valid dissection to cut each element rectangle off the bar in its own piece.\footnote{Because $k = n$, $a_1$ will remain attached to the bar, forcing it to be the first element rectangle placed in the first partition rectangle.  Because the order of and within partitions does not matter, this constraint does not affect the \textsc{5-Partition} simulation.}  When $\delta$ is small enough, the resulting packing problem is a direct simulation of \textsc{5-Partition}.

Intuitively, each dissected piece should contain only one element rectangle. Our reduction sets $d_s$ large enough that any piece containing parts of two element rectangles does not fit in a partition rectangle. At the same time, we pick $d_t$ large enough that no piece can be placed in more than one partition rectangle.  Thus one could plausibly prove that each element rectangle must be in its own piece.

Unfortunately, we were unable to prove that each element rectangle must be in its own piece. For each element rectangle, we define the \emph{trimmed element rectangle} corresponding to each element rectangle as the rectangle resulting from ignoring the lower $4\delta$ of the element rectangle's height; see Figure~\ref{fig-trimmedelmrect}. In other words, for each $a_i$, the corresponding trimmed element rectangle is the rectangle that shares the upper left corner with the element rectangle and is of width $a_i$ and height $1 - 4\delta$.

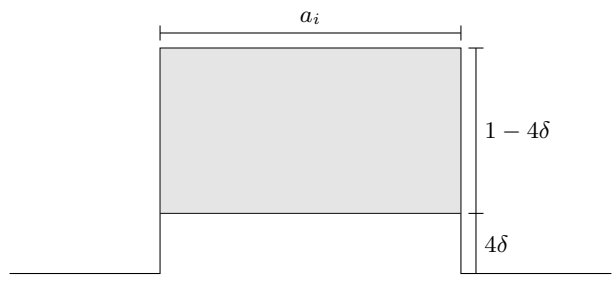
\begin{figure}[h] 
\centering
\begin{tikzpicture}
% upper edge
\draw (1,1)--(3,1)--(3,4)--(7,4)--(7,1)--(9,1);
% lower edge
\draw (1,0.8)--(9,0.8);
% trimmed elm rectangle edge
\draw (3, 1.8)--(7,1.8);
% shading trimmed elm rectangle
\path [fill=gray, opacity=0.2] (3,1.8) rectangle (7,4);
% measurement lines
\draw (3,4.2)--(7,4.2);
\draw (3,4.3)--(3,4.1);
\draw (7,4.3)--(7,4.1);
\draw (7.2,1)--(7.2,4);
\draw (7.1,4)--(7.3,4);
\draw (7.1,1.8)--(7.3,1.8);
% text
\node[above] at (5,4.2) {$a_i$};
\node[right] at (7.2,1.4) {$4\delta$};
\node[right] at (7.2,2.9) {$1 - 4\delta$};
\end{tikzpicture}
\caption{The $i$th trimmed element rectangle. \label{fig-trimmedelmrect}}
\end{figure}

While we could not prove that each element rectangle is in its own piece, we can prove the corresponding statement about trimmed element rectangles:

\begin{lemma} \label{lem-one-per-rec}
	If $P$ can be cut into pieces that can be packed into $Q$, then each of these pieces intersect with at most one trimmed element rectangle.
\end{lemma}

The proofs of both of our main theorems use this lemma. The intuition behind the proof of this lemma is similar to the intuitive argument for why each element rectangle should be in its own piece. As the details of the proof are not central to this paper, we defer the proof of this lemma to \ifabstract Appendix\else Section\fi~\ref{sec-lem-proof}.

\section{Proof of NP-hardness of \textsc{$k$-Piece Dissection}} \label{sec-proof}

\iffull
	Before we prove Theorem~\ref{thm-main}, we state the result from~\cite{Garey-Johnson-1979} for \textsc{5-Partition}:
	\begin{theorem}[\cite{Garey-Johnson-1979}] \label{thm-5-part-hard}
		\textsc{5-Partition} is NP-hard.\footnote{As stated earlier, the result from~\cite{Garey-Johnson-1979} is for \textsc{4-Partition}, but \textsc{4-Partition} is easily reduced to \textsc{5-Partition}; see Section~\ref{sec-prob}.}
	\end{theorem}
\fi

We now prove Theorem~\ref{thm-main}.

\begin{proof}[of Theorem~\ref{thm-main}]
We prove that the reduction described in the previous section
is indeed a valid reduction from \textsc{5-Partition}.
The reduction clearly runs in polynomial time. We are left to prove that the instance of \textsc{$k$-Piece Dissection} produced by the reduction is a \textsc{yes} instance if and only if the input \textsc{5-Partition} is also a \textsc{yes} instance.

\paragraph{(\textsc{5-Partition}$\implies$\textsc{$k$-Piece Dissection})} Suppose that the \textsc{5-Partition} instance is a \textsc{yes} instance. Given a \textsc{5-Partition} solution, we can cut all but the first element rectangle off the bar and pack them in the partition rectangles according to the \textsc{5-Partition} solution. The piece containing the first element rectangle must be placed at the very left of the first partition rectangle, but we can reorder the partitions in the \textsc{5-Partition} solution so that the first element is in the first partition. As a result, the \textsc{$k$-Piece Dissection} instance is also a \textsc{yes} instance.

\paragraph{(\textsc{$k$-Piece Dissection}$\implies$\textsc{5-Partition})} Suppose that the \textsc{$k$-Piece Dissection} instance is a \textsc{yes} instance, i.e., $P$ can be cut into $k$ pieces that can then be placed into $Q$. By Lemma~\ref{lem-one-per-rec}, no two trimmed element rectangles are in the same piece. Because there are $n = k$ such rectangles, each piece contains exactly one whole trimmed element rectangle. Because these pieces can be packed into $Q$, we must also be able to pack all the trimmed element rectangles into $Q$ (with some space in $Q$ left over).

Let $B_i$ be the set of all trimmed element rectangles (in the packing configuration) that intersect the $i$th partition rectangle. From our choice of $d_t$, each trimmed element rectangle can intersect with at most one partition rectangle. Moreover, no trimmed element rectangles fit entirely in the bar area, so each of them must intersect with at least one partition rectangle. This means that $B_1, \dots, B_{n/5}$ is a partition of the set of all trimmed element rectangles. Let $A_i$ be the set of all integers in $A$ corresponding to the trimmed element rectangles in $B_i$. Observe that $A_1, \dots, A_{n/5}$ is a partition of $A$.

We claim that $A_1, \dots, A_{n/5}$ is indeed a solution for \textsc{5-Partition}. Assume for the sake of contradiction that $A_1, \dots, A_{n/5}$ is not a solution, that is, $\sum_{a \in A_i} a \ne p$ for some $i$. Because $\sum_{a \in A} a = p(n/5)$, there exists $j$ such that $\sum_{a \in A_j} a > p$. Because all $a \in A$ are integers and $p$ is an integer, $\sum_{a \in A_j} a \geq p + 1$.

Consider the $j$th partition rectangle. Define the \emph{extended partition rectangle} as the area that includes a partition rectangle, the bar area directly below it, and the bar $\delta/2$ to the left and to the right of the partition rectangle. Figure~\ref{fig-extendedpartrec} illustrates an extended partition rectangle.

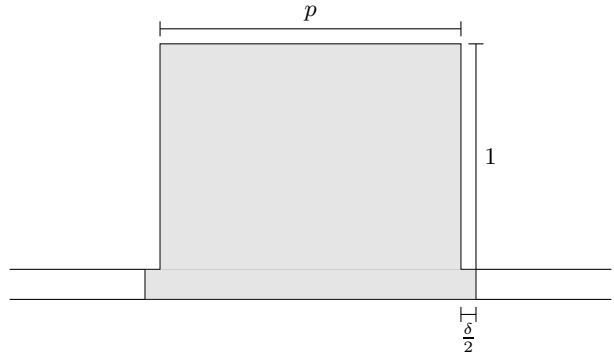
\begin{figure}[h] 
\centering
\begin{tikzpicture}
% upper edge
\draw (1,1)--(3,1)--(3,4)--(7,4)--(7,1)--(9,1);
% lower edge
\draw (1,0.6)--(9,0.6);
% shading extended partition rectangle
\draw (2.8,0.6)--(2.8,1);
\draw (7.2,0.6)--(7.2,1);
\path [fill=gray, opacity=0.2] (3,1) rectangle (7,4);
\path [fill=gray, opacity=0.2] (2.8,0.6) rectangle (7.2,1);
% measurement lines
\draw (3,4.2)--(7,4.2);
\draw (3,4.3)--(3,4.1);
\draw (7,4.3)--(7,4.1);

\draw (7.2,1)--(7.2,4);
\draw (7.1,4)--(7.3,4);

\draw (7,0.4)--(7.2,0.4);
\draw (7,0.3)--(7,0.5);
\draw (7.2,0.3)--(7.2,0.5);
% text
\node[right] at (7.2,2.5) {$1$};
\node[above] at (5,4.2) {$p$};
\node[below] at (7.1,0.4) {$\frac{\delta}{2}$};
\end{tikzpicture}
\caption{The shaded area is the extended partition rectangle corresponding to this partition rectangle. \label{fig-extendedpartrec}}
\end{figure}

Consider any trimmed element rectangle in the packing configuration that intersects with this partition rectangle. We claim that each such trimmed element rectangle must be wholly contained in the extended partition rectangle.

Consider the area of the trimmed element rectangle outside the partition rectangle and the bar below it. If this is not empty, this must be a right triangle with hypotenuse on the extension down to the bar of a vertical side of the partition rectangle (see Figure~\ref{fig-trimmedinextended}). The hypotenuse of this triangle is of length at most $\delta$, so the height of the triangle (perpendicular to the hypotenuse) is at most $\delta/2$. Thus, the triangle must be in the extended partition rectangle. Thus the whole trimmed element rectangle must be in the extended partition rectangle, as claimed.

\begin{figure}[h] 
\centering
\begin{tikzpicture}[scale=0.7]
% upper edge
\draw (4,2)--(6,2)--(6,7)--(12,7)--(12,2)--(14,2);
% lower edge
\draw (4,1.2)--(14,1.2);
% extension of left edge
\draw (6,1.2)--(6,2);
% trimmed rectangle
\draw[fill=gray, opacity=0.2] (7,1.2) -- ++(165:1.17) -- ++(75:5) -- ++(345:1.17) -- ++(255:5);
\draw (7,1.2) -- ++(165:1.17) -- ++(75:5) -- ++(345:1.17) -- ++(255:5);
% measurement lines
\draw (6,7.4)--(12,7.4);
\draw (6,7.6)--(6,7.2);
\draw (12,7.6)--(12,7.2);

\draw (12.4,2)--(12.4,7);
\draw (12.2,7)--(12.6,7);

\draw (14.4,2)--(14.4,1.2);
\draw (14.2,2)--(14.6,2);
\draw (14.2,1.2)--(14.6,1.2);
% text
\node[above] at (10,7.4) {$p$};
\node[right] at (12.4,4.5) {$1$};
\node[right] at (14.4,1.6) {$\delta$};
\end{tikzpicture}
\caption{A configuration where the trimmed element rectangle is outside of the partition rectangle and the bar below it. The shaded area is the trimmed element rectangle. The area of the trimmed element rectangle outside of the partition rectangle and the bar directly below it is a right triangle with hypotenuse on the extension of a vertical edge of the partition rectangle. This extension is of length $\delta$. \label{fig-trimmedinextended}}
\end{figure}
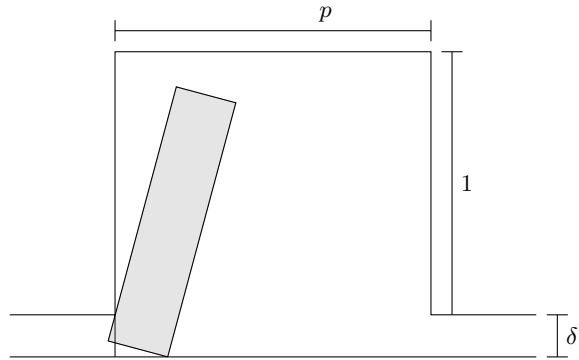

The area of the extended partition rectangle is $p + p\delta + \delta^2 < p + 1/2$. However, the total area of the trimmed element rectangles contained in this area is $\sum_{a \in A_j} a (1 - 4\delta) = \sum_{a \in A_j} a - 4\delta\sum_{a \in A_j}a \geq (p + 1) - 4\delta\sum_{a \in A_j}a > p + 1/2$, which is a contradiction.

Thus we conclude that $A_1, \dots, A_{n/5}$ is a solution to \textsc{5-Partition}, which implies that the \textsc{5-Partition} instance is a \textsc{yes} instance as desired.
%
%Thus, from Theorem~\ref{thm-5-part-hard}, we can conclude that
%\textsc{$k$-Piece Dissection} is also NP-hard.
%Because 5-Partition is NP-hard, so is \textsc{$k$-Piece Dissection}.
\end{proof}

\section{Proof of Inapproximability of \textsc{Min Piece Dissection}} \label{sec-proof-approx}

In this section, we show the inapproximability of \textsc{Min Piece Dissection} via a reduction from the intermediate problem \textsc{Max 5-Partition}.
In Section~\ref{subsec-approx-m5p} below, we prove the following
inapproximability result for \textsc{Max 5-Partition}:

\begin{lemma} \label{lem-5-part-approx}
	There is a constant $\alpha_{\MFP} > 1$ such that \textsc{Gap}\textsubscript{\textsc{Max-5-Partition}}$[n(1-\varepsilon)/5, n(1/\alpha_{\MFP} + \varepsilon)/5]$ is NP-hard for any sufficiently small constant $\varepsilon>0$.\footnote{The best $\alpha_{\MFP}$ we can achieve here is 216/215.}
\end{lemma}

Lemma~\ref{lem-5-part-approx} implies that it is hard to approximate \textsc{Max 5-Partition} to within an $\alpha_{\MFP} - \varepsilon$ ratio for
any sufficiently small $\varepsilon>0$.
The proof of inapproximability of \textsc{Max 5-Partition} largely relies on the reduction used to prove NP-hardness of \textsc{4-Partition} in~\cite{Garey-Johnson-1979},
so we defer the proof of this lemma to \ifabstract Appendix\else Subsection\fi~\ref{subsec-approx-m5p}.
Here we focus on the gap preservation of the reduction, which
implies Theorem~\ref{thm-main-approx}.

\begin{lemma} \label{lem-main-approx}
	There is a constant $\alpha_{\MPD} > 1$ such that the following properties hold for the reduction described in Section~\ref{sec-reduction}:
	\begin{itemize}
		\item if the optimum of the \textsc{Max 5-Partition} instance is at least $n(1-\varepsilon)/5$, then the optimum of the resulting \textsc{Min Piece Dissection} instance is at most $n(1+\varepsilon/5)$; and
		\item if the optimum of the \textsc{Max 5-Partition} instance is at most $n(1/\alpha_{\MFP} + \varepsilon)/5$, then the optimum of the resulting \textsc{Min Piece Dissection} is at least $n(\alpha_{\MPD} + \varepsilon/5)$.
	\end{itemize}
\end{lemma}

Because it is NP-hard to distinguish the two cases of the input \textsc{Max 5-Partition} instance, it is also NP-hard to approximate \textsc{Min Piece Dissection} to within an $\alpha_{\MPD} - \varepsilon$ ratio for any
sufficiently small constant $\varepsilon>0$.
Thus, Lemma~\ref{lem-main-approx} immediately implies Theorem~\ref{thm-main-approx}.
It remains to prove Lemma~\ref{lem-main-approx}:

\begin{proof}[of Lemma~\ref{lem-main-approx}]
	We will show that both properties are true when we choose $\alpha_{\MPD}$ to be $1 + (1 - 1/\alpha_{\MFP})/5$.

	\paragraph{(\textsc{Max 5-Partition}$\implies$\textsc{Min Piece Dissection})} Suppose that the input \textsc{Max 5-Partition} instance has optimum at least $n(1-\varepsilon)/5$. Let $A_1, \dots, A_{m'}$ be the optimal partition where $m' \geq n(1-\varepsilon)/5$. We cut $P$ into pieces as follows:
	\begin{enumerate}
		\item First, we cut every element rectangle except the first one from the bar.
		\item Next, let the indices of the elements in $A - (A_1 \cup A_2 \cup \cdots \cup A_{m'})$ be $i_1, \dots, i_{l}$ where $1 \leq i_1 < i_2 < \cdots < i_{l} \leq n$.
		\item For each $i = 1, \dots, n/5 - m'$, let $j$ be the smallest index such that $a_{i_1} + \cdots + a_{i_j} \geq ip$. Cut the piece corresponding to $a_{i_j}$ vertically at position $ip - \left(a_{i_1} + \cdots + a_{i_{j - 1}}\right)$ from the left. (If the intended cut position is already the right edge of the piece, then do nothing.)
	\end{enumerate}

	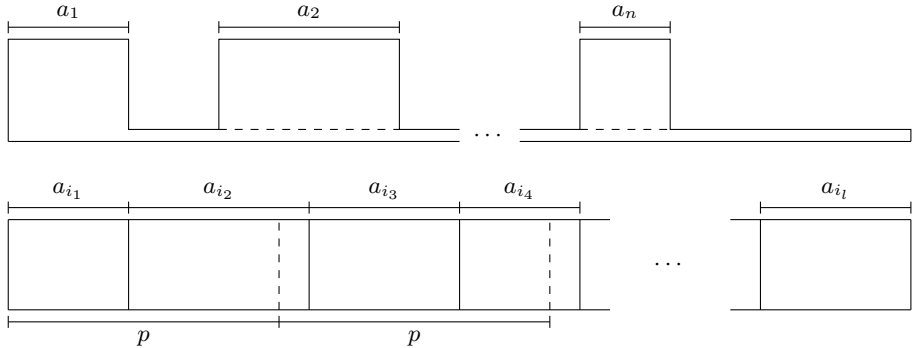
\begin{figure}[h] 
	\centering
	\begin{tikzpicture}[scale=0.8]
	% cuts in the first step
	% upper edge
	\draw (0,3.8)--(0,5.5)--(2,5.5)--(2,4)--(3.5,4)--(3.5,5.5)--(6.5,5.5)--(6.5,4)--(7.5,4);
	\draw (8.5,4)--(9.5,4)--(9.5,5.5)--(11,5.5)--(11,4)--(11.5,4)--(15,4)--(15,3.8);
	% lower edge
	\draw (0,3.8)--(7.5,3.8);
	\draw (8.5,3.8)--(15,3.8);
	% measurement line
	\draw (0,5.7)--(2,5.7);
	\draw (0,5.6)--(0,5.8);
	\draw (2,5.6)--(2,5.8);
	\draw (3.5,5.7)--(6.5,5.7);
	\draw (3.5,5.6)--(3.5,5.8);
	\draw (6.5,5.6)--(6.5,5.8);
	\draw (9.5,5.7)--(11,5.7);
	\draw (9.5,5.6)--(9.5,5.8);
	\draw (11,5.6)--(11,5.8);
	% cuts
	\draw[dashed] (3.5,4)--(6.5,4);
	\draw[dashed] (9.5,4)--(11,4);
	% text
	\node[above] at (1,5.7) {$a_1$};
	\node[above] at (5,5.7) {$a_2$};
	\node[above] at (10.25,5.7) {$a_n$};
	\node at (8,3.9) {$\cdots$};

	% cuts in the second step
	% lower edge
	\draw (0,1)--(10,1);
	\draw (12,1)--(15,1);
	% upper edge
	\draw (0,2.5)--(10,2.5);
	\draw (12,2.5)--(15,2.5);
	% vertical edges
	\draw (0,1)--(0,2.5);
	\draw (2,1)--(2,2.5);
	\draw (5,1)--(5,2.5);
	\draw (7.5,1)--(7.5,2.5);
	\draw (9.5,1)--(9.5,2.5);
	\draw (12.5,1)--(12.5,2.5);
	\draw (15,1)--(15,2.5);
	% cuts
	\draw[dashed] (4.5,1)--(4.5,2.5);
	\draw[dashed] (9,1)--(9,2.5);

	% measurement lines
	\draw (0,2.7)--(9.5,2.7);
	\draw (12.5,2.7)--(15,2.7);
	\draw (0,0.8)--(9,0.8);
	\draw (0,2.6)--(0,2.8);
	\draw (2,2.6)--(2,2.8);
	\draw (5,2.6)--(5,2.8);
	\draw (7.5,2.6)--(7.5,2.8);
	\draw (9.5,2.6)--(9.5,2.8);
	\draw (12.5,2.6)--(12.5,2.8);
	\draw (15,2.6)--(15,2.8);
	\draw (0,0.9)--(0,0.7);
	\draw (4.5,0.9)--(4.5,0.7);
	\draw (9,0.9)--(9,0.7);

	% text
	\node[above] at (1, 2.7) {$a_{i_1}$};
	\node[above] at (3.5, 2.7) {$a_{i_2}$};
	\node[above] at (6.25, 2.7) {$a_{i_3}$};
	\node[above] at (8.5, 2.7) {$a_{i_4}$};
	\node[above] at (13.75, 2.7) {$a_{i_l}$};
	\node[below] at (2.25,0.8) {$p$};
	\node[below] at (6.75,0.8) {$p$};
	\node at (11,1.75) {$\cdots$};
	\end{tikzpicture}
	\caption{An illustration of how the source polygon $P$ is cut. The cuts from step 1 are shown as dashed lines on the top figure; every element rectangle except the first one is cut from the bar. On the bottom, the cuts from step 3 are demonstrated. We can think of the cutting process as first arranging $a_{i_1}, \dots, a_{i_l}$ consecutively and then cutting at $p, 2p, \dots$. \label{fig-cut}}
	\end{figure}

	To pack these pieces into $Q$, we arrange all pieces whose corresponding elements are in partitions in the optimal \textsc{Max 5-Partition} solution, then pack the remaining pieces into the remaining partition rectangles using the additional cuts made in step 3.  We leave the piece containing the first element rectangle (and the bar) at its position in $P$, but this does not constrain our solution because the other pieces and the partitions can be freely reordered.

	The number of cuts in step 1 is $n - 1$ and in step 3 is at most $n/5 - m' \leq \varepsilon n/5$. Thus the total number of cuts is at most $n - 1 + \varepsilon n / 5$, so the number of pieces is at most $1 + (n - 1 + \varepsilon n / 5) = n(1 + \varepsilon/5)$ as desired.

	\paragraph{(\textsc{Min Piece Dissection}$\implies$\textsc{Max 5-Partition})} We prove this property in its contrapositive form. Suppose that the resulting \textsc{Min Piece Dissection} has an optimum of $k < n(\alpha_{\MPD} + \varepsilon/5)$. Let us call these $k$ pieces $R_1, \dots, R_k$.

	For each $i = 1, \dots, k$, let $R'_i$ denote the intersection between $R_i$ with the union of all trimmed element rectangles. By Lemma~\ref{lem-one-per-rec}, each trimmed element rectangle can intersect with only one piece. This means that each $R'_i$ is a part of a trimmed element rectangle. (Note that $R'_i$ can be empty; in this case, we say that it belongs to the first trimmed element rectangle.)

	Consider $R'_1, \dots, R'_k$. Because each of them is a part of a trimmed rectangle and there are $n$ trimmed rectangles, at most $k - n$ trimmed rectangles contain more than one of the $R'_i$. In other words, there are at least $n - (k - n) = 2n - k$ indices $i$ such that $R'_i$ is a whole trimmed element rectangle. Without loss of generality, suppose that $R'_1, \dots, R'_{2n-k}$ are entire trimmed element triangles. 

	We call a partition rectangle a \emph{good partition rectangle} if it does not intersect with any of $R'_{2n-k+1}, \dots, R'_{n}$ in the packing configuration. From our choice of $d_t$, each $R'_i$ which is part of a trimmed element rectangle can intersect with at most one partition rectangle. As a result, there are at least $n/5 - (k - n)$ good partition rectangles.

	For each good partition rectangle $O$, let $A_O$ be the subset of all elements of $A$ corresponding to $R'_i$s that intersect $O$. (Because $O$ is a good partition rectangle, each $R'_i$ that intersects $O$ is always a whole trimmed element rectangle.)

	We claim that the collection of $T_O$'s for all good partition rectangles $O$ is a solution to the \textsc{Max 5-Partition} instance. We will show that this is indeed a valid solution. First, observe again that, because each $R'_i$ intersects with at most one partition rectangle, all $A_O$'s are mutually disjoint. Thus, we now only need to prove that the sum of elements of $A_O$ is exactly the target sum $p$.

	Suppose for the sake of contradiction that there exists a good partition rectangle $O$ such that $\sum_{a \in A_O} a \ne p$. Consider the following two cases.
	\begin{description}
		\item[Case 1:] $\sum_{a \in A_O} a > p$.
\medskip

		As we showed in the proof of Theorem~\ref{thm-main}, each trimmed element rectangle corresponding to $a \in A_O$ must be in the extended partition rectangle.  By an argument similar to the argument used in the proof of Theorem~\ref{thm-main}, the total area of all these trimmed element rectangles is more than the area of the extended partition rectangle, which is a contradiction.
\medskip
		\item[Case 2:] $\sum_{a \in A_O} a < p$.
\medskip

		Because every $a \in A_O$ and $p$ are integers, $\sum_{a \in A_O} a + 1 \leq p$. From the definition of $A_O$, no trimmed element rectangles apart from those in $A_O$ intersect $O$. Hence the total area that trimmed element rectangles contribute to $O$ is at most $$\left(\sum_{a \in A_O}a\right)(1 - 4\delta) < \sum_{a \in A_O}a \leq p - 1.$$ This means that an area of at least $1$ unit square in $O$ is not covered by any of the trimmed element rectangles. However, the area of the source polygon outside of all the trimmed element rectangles is $$\delta\left(\left(\frac{n}{5} - 1\right)d_t + \sum_{a \in A} a\right) + 4 \delta\left(\sum_{a \in A} a\right) < 1,$$ which is a contradiction.
	\end{description}

	Hence, the solution defined above is a valid solution. Because the number of good partition rectangles is at least $n/5 - (k - n) > n/5 - n(\alpha_{\MPD} + \varepsilon/5 - 1) = n(1/\alpha_{\MFP} - \varepsilon)/5$, the solution contains more than $n(1/\alpha_{\MFP} - \varepsilon)/5$ subsets, which completes the proof of the second property.
\end{proof}

\later{
\ifabstract
	\section{Proof of NP-hardness of Approximation of \textsc{Max 5-Partition}}
\else
	\section{Proof of NP-hardness of Approximation of \textsc{Max 5-Partition}}
\fi
\label{subsec-approx-m5p} 

To prove hardness of approximating \textsc{Max 5-Partition}, we will reduce from \textsc{4-Uniform 4-Dimensional Matching}. 

\textsc{2-Uniform 3-Dimensional Matching}, a problem similar to \textsc{4-Uniform 4-Dimensional Matching}, was first proved to be hard to approximate by Chleb{\'{\i}}k and Chleb{\'{\i}}kov{\'{a}}~\cite{Chlebik-Chlebikova-2002}. Shortly after, \textsc{4-Uniform 4-Dimensional Matching} was proved to be NP-hard to approximate by Berman and Karpinski~\cite{Berman-Karpinski-2003}. These two results, however, are not strong enough for our proof.\footnote{Our proof requires the \emph{near perfect completeness} property. This property refers to the value $\beta$ for which \textsc{Gap}\textsubscript{\textsc{4DM}}[$\beta, \gamma$] is NP-hard. To have near perfect completeness, $\beta$ needs to correspond to an almost perfect matching.} For the purpose of this paper, we will use the following inapproximability result for \textsc{4-Uniform 4-Dimensional Matching} by Hazan, Safra and Schwartz~\cite{Hazan-Safra-Schwartz-2003}:

\begin{theorem}[\cite{Hazan-Safra-Schwartz-2003}] \label{thm-match}
There is a constant $\alpha_{\FDM} > 1$ such that, for
any sufficiently small constant $\varepsilon>0$,
\textsc{Gap}\textsubscript{\textsc{4DM}}$[n'(1 - \varepsilon), n'(1/\alpha_{\FDM} + \varepsilon)]$ is NP-hard where $n'$ denote $|V^1| = |V^2| = |V^3| = |V^4|$.\footnote{In~\cite{Hazan-Safra-Schwartz-2003}, $\alpha_{\FDM}$ is 54/53.}
\end{theorem}

We now prove Lemma~\ref{lem-5-part-approx} using Theorem~\ref{thm-match}.

\begin{proof}[of Lemma~\ref{lem-5-part-approx}]
First we describe the reduction from \textsc{4-Uniform 4-Dimensional Matching} to \textsc{Max 5-Partition}. The reduction is similar to the reduction used to prove NP-hardness of \textsc{4-Partition} in~\cite{Garey-Johnson-1979}. Given an instance $(V^1, V^2, V^3, V^4, E)$ of \textsc{4-Uniform 4-Dimensional Matching}, we construct an instance $A = \{a_1, \dots, a_n\}$ of \textsc{Max 5-Partition} as follows:
\begin{enumerate}
	\item Let $n = 20n'$. Because $(V^1, V^2, V^3, V^4, E)$ is 4-uniform and balanced, $|E| = 4|V^1| = 4|V^2| = 4|V^3| = 4|V^4| = 4n'$, or equivalently, $n = 4|V^1| + 4|V^2| + 4|V^3| + 4|V^4| + |E|$.
	\item For convenience, let $M = 100(n' + 1)^2$ and, for each $m = 1, \dots, 4$, label $V^m$'s vertices from 1 to $n'$.
	\item For each vertex $i = 1, \dots, n'$ of $V^m$, add four elements to the set $A$: one \emph{matching element} of value $M^8 + 4M^{m + 3} + M^{m - 1} i$ and three \emph{dummy elements} of value $M^8 + M^7 + M^6 + M^5 + M^4 + M^{m - 1} i$.
	\item For each edge $(i, j, k, l) \in E$, add an element of value $M^8 - M^3 l - M^2 k - M j - i$ into $A$.
\end{enumerate}

In the resulting instance of \textsc{Max 5-Partition}, the target partition sum is $5M^8 + 4M^7 + 4M^6 + 4M^5 + 4M^4$. It is easy to see that this can be met by only two types of subsets: an edge element and the four corresponding matching elements, and an edge element and the four corresponding dummy elements. % $\{M^8 - M^3 l - M^2 k - M j - i, M^8 + 4M^4 + i, M^8 + 4M^5 + M j, M^8 + 4M^6 + M^2 k, M^8 + 4M^7 + M^3 l\}$ or $\{M^8 - M^3 l - M^2 k - M j - i, M^8 + M^7 + M^6 + M^5 + M^4 + i, M^8 + M^7 + M^6 + M^5 + M^4 + M j, M^8 + M^7 + M^6 + M^5 + M^4 + M^2 k, M^8 + M^7 + M^6 + M^5 + M^4 + M^3 l\}$ for $(i, j, k, l) \in E$.

A subset of the first type corresponds to an edge in a matching in the \textsc{4-Uniform 4-Dimensional Matching} instance because there is only one matching element for each vertex $i$ of $V^m$, preventing two subsets from sharing a vertex; we will call these subsets \emph{matching subsets}. We will call subsets of the other type \emph{dummy subsets}.

This reduction clearly runs in polynomial time. Select $\alpha_{\MFP}$ such that $1/\alpha_{\MFP} = 1/(4\alpha_{\FDM}) + 3/4$. We prove the two following properties of the reduction, which together immediately yield Lemma~\ref{lem-5-part-approx}:
\begin{enumerate}
	\item If the optimum of the \textsc{4-Uniform 4-Dimensional Matching} instance is at least $n'(1 - \varepsilon)$, then the optimum of the resulting \textsc{Max 5-Partition} instance is at least $n(1-\varepsilon)/5$.
	\item If the optimum of the \textsc{4-Uniform 4-Dimensional Matching} instance is at most $n'(1/\alpha_{\FDM} + \varepsilon)$, then the optimum of the resulting \textsc{Max 5-Partition} instance is at most $n(1/\alpha_{\MFP} + \varepsilon/4)/5$.
\end{enumerate}

\paragraph{(\textsc{4-Uniform 4-Dimensional Matching}$\implies$\textsc{Max 5-Partition})} We now prove the first property. Suppose that the input \textsc{4-Uniform 4-Dimensional Matching} instance has optimum at least $n'(1 - \varepsilon)$. We create a solution to the \textsc{Max 5-Partition} instance by first taking the matching subsets corresponding to the matching. Then, for each edge not in the matching, we pick the corresponding dummy subset if possible. (We cannot pick it if all the dummy elements corresponding to a vertex in the edge are already taken.) This is our solution for \textsc{Max 5-Partition}.

Consider each vertex $i \in V^m$. If $i$ is in the matching, then there are only three dummy subsets left, so we will not run out of dummy elements for $i$ when we try to pick the dummy subsets. On the other hand, if $i$ is not in the matching, we will run out of dummy elements for one of the dummy subsets.

Because the matching is of size at least $n'(1 - \varepsilon)$, the number of vertices not in the matching is at most $4n'\varepsilon$. Because each such vertex causes at most one dummy subset to not be picked, at most $4n'\varepsilon$ dummy subsets are not picked. Hence, the solution created for the \textsc{Max 5-Partition} instance has at least $n/5 - 4n'\varepsilon = n(1 - \varepsilon)/5$ subsets, which concludes the proof for the first property.

\paragraph{(\textsc{Max 5-Partition}$\implies$\textsc{4-Uniform 4-Dimensional Matching})} We will show the second property using its contrapositive. Suppose that the \textsc{Max 5-Partition} instance has optimum more than $n(1/\alpha_{\MFP} + \varepsilon/4)/5$, i.e., there exist more than $n(1/\alpha_{\MFP} + \varepsilon)/5$ disjoint subsets each having a sum of elements equal to $5M^8 + 4M^7 + 4M^6 + 4M^5 + 4M^4$. Consider the dummy subsets among these subsets.

Because there are only three dummy elements corresponding to each vertex, there are only $3 \cdot 4n' = 12n'$ dummy elements in $A$. Four dummy elements are needed for each dummy subset, so there are at most $12n'/4 = 3n' = 3n/20$ such subsets. As a result, there are more than $n(1/\alpha_{\MFP} + \varepsilon/4)/5 - 3n/20 = n(1/\alpha_{\FDM} + \varepsilon)/20 = n'(1/\alpha_{\FDM} + \varepsilon)$ matching subsets. These subsets correspond to a matching of size more than $n'(1/\alpha_{\FDM} + \varepsilon)$ in the \textsc{4-Uniform 4-Dimensional Matching} instance. Thus, the second property holds.

Thus, the two properties are true and Lemma~\ref{lem-5-part-approx} follows immediately.
\end{proof}
}

\later{
\section{Proof of Lemma~\ref{lem-one-per-rec}} \label{sec-lem-proof}

In this section, we prove Lemma~\ref{lem-one-per-rec}. Before we do so, we prove the following helper lemma.
\begin{lemma} \label{lem-tri-height}
	For any triangle that can be fit into a bar of height $\delta$, at least one of its heights is of length at most $2\delta$.
\end{lemma}

\begin{proof}
	Consider the leftmost and rightmost vertices, denoted $A$ and $B$ respectively, of the triangle in the configuration that fits into the bar. Consider a strip where $A$ is on the left edge, $B$ is on the right edge and the upper and lower edges are on the edges of the bar. This strip contains the whole triangle and its width is at most $AB$, so the strip has area at most $\delta\cdot AB$. The triangle has area equal to half of $AB$ multiplied by the height with respect to base $AB$, so the height with respect to base $AB$ is at most $2\delta$.
\end{proof}

	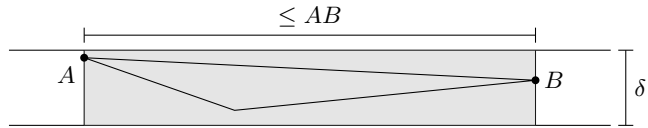
\begin{figure}[h] 
	\centering
	\begin{tikzpicture}
	% upper edge
	\draw (1,1)--(9,1);
	% lower edge
	\draw (1,0)--(9,0);
	% left strip edge
	\draw (2,0)--(2,1);
	% right strip edge
	\draw (8,0)--(8,1);
	% shaded strip
	\path [fill=gray, opacity=0.2] (2,0) rectangle (8,1);
	% the triangle
	\draw (2,0.9)--(4,0.2)--(8,0.6)--(2,0.9);
	\fill (2,0.9) circle (1.5pt);
	\node[below left] at (2,0.9) {$A$};
	\fill (8,0.6) circle (1.5pt);
	\node[right] at (8,0.6) {$B$};
	% measurement lines
	\draw (9.2,0)--(9.2,1);
	\draw (9.1,0)--(9.3,0);
	\draw (9.1,1)--(9.3,1);
	\node[right] at (9.2,0.5) {$\delta$};

	\draw (2,1.2)--(8,1.2);
	\draw (2,1.1)--(2,1.3);
	\draw (8,1.1)--(8,1.3);
	\node[above] at (5,1.2) {$\leq AB$};
	\end{tikzpicture}
	\caption{A triangle fitted into a bar of height $\delta$. The shaded area is the strip considered in the proof of Lemma~\ref{lem-tri-height}. \label{fig-trianinstrip}}
	\end{figure}

We are now ready to prove Lemma~\ref{lem-one-per-rec}.

\begin{proof}[of Lemma~\ref{lem-one-per-rec}]
	Suppose for the sake of contradiction that there exists a piece $R$ that intersects with two different trimmed element rectangles. Let the indices of the trimmed element rectangles be $i$ and $j$ where $i < j$.

	Define area $T_i$ to be the whole area below the $i$th trimmed element rectangle together with the area of the bar of width $d_s/3$ immediately on the right of the element rectangle. Similarly, define $T_j$ to be the whole area below the $j$th trimmed element rectangle together with the area of the bar of width $d_s/3$ immediately on the left of the element rectangle. $T_i$ and $T_j$ are illustrated in Figure~\ref{fig-trimmedex}.

	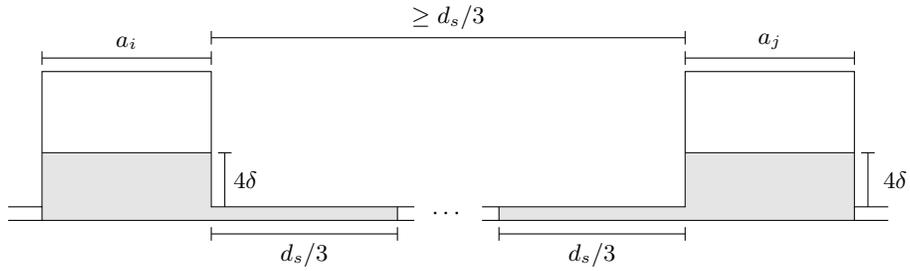
\begin{figure}[h] 
	\centering
	\begin{tikzpicture}[scale=0.9]
	% upper edge
	\draw (2.5,1)--(3,1)--(3,3)--(5.5,3)--(5.5,1)--(8.5,1);
	\draw (9.5,1)--(12.5,1)--(12.5,3)--(15,3)--(15,1)--(15.5,1);
	% lower edge
	\draw (2.5,0.8)--(8.5,0.8);
	\draw (9.5,0.8)--(15.5,0.8);
	% middle edge
	\draw (3,1.8)--(5.5,1.8);
	\draw (12.5,1.8)--(15,1.8);
	
	% shading the below area
	\path [fill=gray, opacity=0.2] (3,0.8)--(3,1.8)--(5.5,1.8)--(5.5,1)--(8.25,1)--(8.25,0.8);
	\path [fill=gray, opacity=0.2] (9.75,0.8)--(9.75,1)--(12.5,1)--(12.5,1.8)--(15,1.8)--(15,0.8);
	% draw edges to the shaded area
	\draw (3,0.8)--(3,1);
	\draw (8.25,1)--(8.25,0.8);
	\draw (9.75,0.8)--(9.75,1);
	\draw (15,0.8)--(15,1);

	% measurement lines
	\draw (3,3.2)--(5.5,3.2);
	\draw (3,3.3)--(3,3.1);
	\draw (5.5,3.3)--(5.5,3.1);
	\node[above] at (4.25,3.2) {$a_i$};

	\draw (5.5,3.5)--(12.5,3.5);
	\draw (5.5,3.4)--(5.5,3.6);
	\draw (12.5,3.4)--(12.5,3.6);
	\node[above] at (9,3.5) {$\geq d_s/3$};

	\draw (12.5,3.2)--(15,3.2);
	\draw (12.5,3.3)--(12.5,3.1);
	\draw (15,3.3)--(15,3.1);
	\node[above] at (13.75,3.2) {$a_j$};

	\draw (5.7,1)--(5.7,1.8);
	\draw (5.6,1.8)--(5.8,1.8);
	\node[right] at (5.7,1.4) {$4\delta$};

	\draw (15.2,1)--(15.2,1.8);
	\draw (15.1,1.8)--(15.3,1.8);
	\node[right] at (15.3,1.4) {$4\delta$};

	\draw (5.5,0.6)--(8.25,0.6);
	\draw (5.5,0.5)--(5.5,0.7);
	\draw (8.25,0.5)--(8.25,0.7);
	\node[below] at (6.875,0.6) {$d_s/3$};

	\draw (9.75,0.6)--(12.5,0.6);
	\draw (9.75,0.5)--(9.75,0.7);
	\draw (12.5,0.5)--(12.5,0.7);
	\node[below] at (11.125,0.6) {$d_s/3$};

	\node at (9, 0.9) {$\cdots$};
	\end{tikzpicture}
	\caption{$T_i$ and $T_j$ are shown as the shaded areas on the left and on the right respectively. \label{fig-trimmedex}}
	\end{figure}

	We first prove the following claim: in the final packing configuration in $Q$, at least one point in $R \cap T_i$ must be in a partition rectangle.

	Suppose for the sake of contradiction that $R \cap T_i$ is completely contained in the bar. For convenience, we name the points as follows:
	\begin{itemize}
		\item Let $C$ be the lower left corner of the $i$th trimmed element rectangle.
		\item Let $D$ be the lower right corner of the trimmed element rectangle.
		\item Let $E$ be the lower right corner of the $i$th (untrimmed) element rectangle.
		\item Let $F$ be the intersection of the extension of the right edge of the element rectangle with the lower edge of the bar.
		\item Let $G$ be the point on the upper edge of the bar that is distance $d_s/3$ to the right of $E$.
		\item Let $H$ be the point on the lower edge of the bar that is distance $d_s/3$ to the right of $F$.
	\end{itemize}
	These points are marked in Figure~\ref{fig-pointsname}.
	
	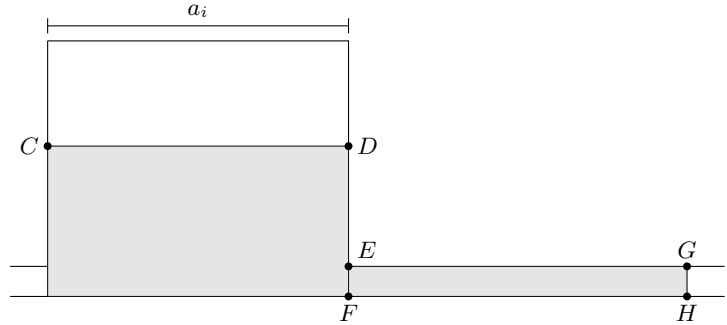
\begin{figure}[h] 
	\centering
	\begin{tikzpicture}
	% upper edge
	\draw (2.5,1)--(3,1);
	\draw (3,0.6)--(3,4);
	\draw (3,4)--(7,4)--(7,0.6);
	\draw (7,1)--(12,1);
	% lower edge
	\draw (2.5,0.6)--(12,0.6);
	% middle edge
	\draw (3,2.6)--(7,2.6);
	% shading extended partition rectangle
	\path [fill=gray, opacity=0.2] (3,0.6)--(3,2.6)--(7,2.6)--(7,1)--(11.5,1)--(11.5,0.6);
	\draw (11.5,0.6)--(11.5,1);
	% measurement lines
	\draw (3,4.2)--(7,4.2);
	\draw (3,4.3)--(3,4.1);
	\draw (7,4.3)--(7,4.1);
	\node[above] at (5,4.2) {$a_i$};
	% dots
	\fill (3,2.6) circle (1.5pt);
	\node[left] at (3,2.6) {$C$};
	\fill (7,2.6) circle (1.5pt);
	\node[right] at (7,2.6) {$D$};
	\fill (7,1) circle (1.5pt);
	\node[above right] at (7,1) {$E$};
	\fill (7,0.6) circle (1.5pt);
	\node[below] at (7,0.6) {$F$};
	\fill (11.5,1) circle (1.5pt);
	\node[above] at (11.5,1) {$G$};
	\fill (11.5,0.6) circle (1.5pt);
	\node[below] at (11.5,0.6) {$H$};
	\end{tikzpicture}
	\caption{An illustration of the points $C, D, E, F, G$ and $H$. The shaded area is the area $T_i$ defined earlier. \label{fig-pointsname}}
	\end{figure}

	From the assumptions that $R$ contains pieces from both trimmed element rectangles and that $R$ is connected, we know the piece must intersect $\overline{CD}, \overline{GH}$ and $\overline{EF}$. Let the points at which $R$ intersects these line segments be $I, J$ and $K$ respectively. By our assumption that $R \cap T_i$ is completely contained in the bar in the packing configuration, $I, J$ and $K$ must all be in the bar in the final configuration. Let $L$ be the intersection of $\overline{IJ}$ and $\overline{DF}$ and $M$ be the intersection between $\overline{IJ}$ and $\overline{EG}$ (see Figure~\ref{fig-inter}; we will show that $L$ lies on $\overline{DE}$ and thus $\overline{IJ}$ indeed intersects $\overline{EG}$). If $I = D$, let $L$ be $D$. Similarly, if $J = G$, let $M$ be $G$. Because $I$ and $J$ are in the bar in the final configuration by our assumption, and the bar is convex, $L$ and $M$ must also be in the bar.

	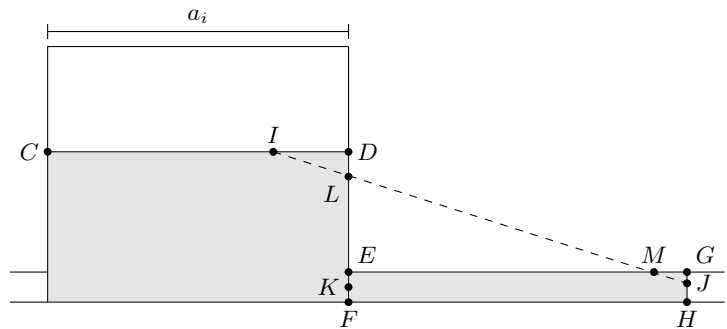
\begin{figure}[h] 
	\centering
	\begin{tikzpicture}
	% upper edge
	\draw (2.5,1)--(3,1);
	\draw (3,0.6)--(3,4)--(7,4);
	\draw (7,4)--(7,0.6);
	\draw (7,1)--(12,1);
	% lower edge
	\draw (2.5,0.6)--(12,0.6);
	% middle edge
	\draw (3,2.6)--(7,2.6);
	% shading extended partition rectangle
	\path [fill=gray, opacity=0.2] (3,0.6)--(3,2.6)--(7,2.6)--(7,1)--(11.5,1)--(11.5,0.6);
	\draw (11.5,0.6)--(11.5,1);
	% measurement lines
	\draw (3,4.2)--(7,4.2);
	\draw (3,4.3)--(3,4.1);
	\draw (7,4.3)--(7,4.1);
	\node[above] at (5,4.2) {$a_i$};
	% dots
	\fill (3,2.6) circle (1.5pt);
	\node[left] at (3,2.6) {$C$};
	\fill (7,2.6) circle (1.5pt);
	\node[right] at (7,2.6) {$D$};
	\fill (7,1) circle (1.5pt);
	\node[above right] at (7,1) {$E$};
	\fill (7,0.6) circle (1.5pt);
	\node[below] at (7,0.6) {$F$};
	\fill (11.5,1) circle (1.5pt);
	\node[above right] at (11.5,1) {$G$};
	\fill (11.5,0.6) circle (1.5pt);
	\node[below] at (11.5,0.6) {$H$};
	\fill (6,2.6) circle (1.5pt);
	\node[above] at (6,2.6) {$I$};
	\fill (11.5,0.85) circle (1.5pt);
	\node[right] at (11.5,0.85) {$J$};
	\draw [dashed] (6,2.6)--(11.5,0.85);
	\fill (7,2.27) circle (1.5pt);
	\node[below left] at (7,2.27) {$L$};
	\fill (11.06,1) circle (1.5pt);
	\node[above] at (11.06,1) {$M$};
	\fill (7,0.8) circle (1.5pt);
	\node[left] at (7,0.8) {$K$};
	\end{tikzpicture}
	\caption{An illustration of the points $L$ and $M$ defined above. \label{fig-inter}}
	\end{figure}

	Consider $\angle LID$. We have $$\tan \angle LID = \frac{LD}{ID} = \frac{ED+GJ}{ID+EG}.$$
	Thus, $$LD = \frac{ED + GJ}{1 + EG/ID}.$$

	Because $ED = 4\delta, EG = d_s/3, ID \leq CD = a_i < d_s/12$ and $GJ \leq GH = \delta$, we have
	\begin{align*}
		LD &< \frac{4\delta + \delta}{1 + (d_s/3)/(d_s/12)} \leq \delta.
	\end{align*}
	Thus $LE = 4\delta - LD > 3\delta$.

	Consider $\angle MLE$. We again have $$\tan \angle MLE = \frac{ME}{LE} = \frac{EG}{LE + GJ}.$$
	Thus $$ME = \frac{EG}{1 + GJ/LE}.$$

	Substituting $EG = d_s/3, GJ \leq GH = \delta$ and $LE > 3\delta$ into the equality above, we have
	\begin{align*}
		ME &> \frac{d_s/3}{1 + \delta/(3\delta)} = \frac{d_s}{12}.
	\end{align*}

	Because $L$ and $K$ are in the bar and $E$ is on $LK$, $E$ must also be in the bar. Consider the triangle $\triangle LEM$. Because all of its vertices are in the bar, the whole triangle is contained in the bar. However, the smallest height of the triangle (the one with the hypotenuse as the base) is of height
	\begin{align*}
		\frac{LE \cdot ME}{LM} &= \frac{LE \cdot ME}{\sqrt{LE^2 + ME^2}} = \frac{LE}{\sqrt{(LE/ME)^2 + 1}} > \frac{LE}{\sqrt{2}} > \frac{3\delta}{\sqrt{2}} > 2\delta
	\end{align*}
	where the first inequality follows from the fact that $LE \leq 4\delta < d_s/12 < ME$.
	This contradicts Lemma~\ref{lem-tri-height}. Hence at least one point in $R \cap T_i$ must be in a partition rectangle in the final packing configuration.

	Similarly, at least one point in $R \cap T_j$ must be in a partition rectangle. Consider a point $P$ in $R \cap T_i$ and a point $Q$ in $R \cap T_j$ that are contained in partition rectangles. Because $P \in T_i$ and $Q \in T_j$, we have $PQ \geq d_s/3$. By our construction, $d_s/3 > \sqrt{p^2 + 1}$, and thus $P$ and $Q$ cannot fit within the same partition rectangle.

	Consider the locations of $P$ and $Q$ in the source polygon. Because $P \in T_i$ and $Q \in T_j$, their $x$ coordinates can differ by at most $(n - 1)d_s + \sum_{a \in A} a$; this corresponds to the extreme case where $P$ is on the left edge of the first element rectangle and $Q$ is on the right edge of the $n$th element rectangle. The $y$ coordinate of the two points also cannot differ by more than $5\delta$. As a result, the distance between the two points is at most $\sqrt{((n - 1)d_s + \sum_{a \in A} a)^2 + (5\delta)^2}$, which is less than $d_t$. Thus, $P$ and $Q$ cannot be in different partition rectangles either. This is a contradiction, completing the proof of Lemma~\ref{lem-one-per-rec}.
\end{proof}
}

%\section{Variations of \textsc{$k$-Piece Dissection} and Open Questions} \label{sec-variant} \label{sec-open}
\section{Variations and Open Questions} \label{sec-variant} \label{sec-open}

Table~\ref{variations table} lists variations of \textsc{$k$-Piece Dissection} and whether our proofs of NP-hardness and inapproximability continue to hold. Because it is obvious from the proofs, we do not give detailed explanations as to why the proofs still work (or do not work) in these settings.  Specifically:
\begin{enumerate}
	\item Our proofs remain valid when the input polygons are restricted to be simple (hole-free) and orthogonal with all edges having integer length.\footnote{Our reduction uses rational lengths, but the polygons can be scaled up to use integer lengths while still being of polynomial size.}
	\item Our results still hold under any cuts that leave each piece connected and Lebesgue measurable.
	\item Our proofs work whether or not rotations and/or reflections are allowed when packing the pieces into~$Q$.
\end{enumerate}

\begin{table}[h!]
	\centering
	\begin{tabular}{| c | c | c |}
		\hline
		\bf Variation on & \bf Variation description & \bf Do our results hold? \\
		\hline \hline
		\multirow{4}{*}{Input Polygons} & Polygons must be orthogonal & YES \\ \cline{2-3}
		& Polygons must be simple (hole-free) & YES \\ \cline{2-3}
		& Edges must be of integer length & YES \\ \cline{2-3}
		& Polygons must be convex & NO \\
		\hline \hline
		\multirow{4}{*}{Cuts Allowed} & Cuts must be straight lines & YES \\ \cline{2-3}
		& Cuts must be orthogonal & YES \\ \cline{2-3}
		& Pieces must be simple (hole-free) & YES \\ \cline{2-3}
		& Pieces may be disconnected & NO \\
		\hline \hline
		\multirow{2}{*}{Packing Rules} & Rotations are forbidden & YES \\ \cline{2-3}
		& Reflections are forbidden & YES \\ \cline{2-3}
		\hline
	\end{tabular}
        \caption{Variations on the dissection problem.}
        \label{variations table}
\end{table}

While we have proved that the \textsc{$k$-Piece Dissection} is NP-hard and that its optimization counterpart is NP-hard to approximate, we are far from settling the complexity of these problems and their variations. We pose a few interesting remaining open questions:
\begin{itemize}
	\item Is \textsc{$k$-Piece Dissection} in NP, or even decidable? We do not know the answer to this question even when only orthogonal cuts are allowed and rotations and reflections are forbidden.
In particular, there exist two-piece orthogonal (staircase) dissections
between pairs of rectangles which seem to require a cut comprised of
arbitrarily many line segments \cite[p.~60]{Frederickson-1997}.
%Precision issues make finding a polynomial-size witness for the problem challenging. In order to overcome this issue, we likely need to nail down where the cuts can be without changing the answer to the problem.

	\iffull
If we require each piece to be a polygon with a polynomial number of sides, then problem becomes decidable.  In fact, we can place this special case in the complexity class $\exists \mathbb{R}$, that is, deciding true sentences of the form $\exists x_1 : \cdots : \exists x_m : \varphi(x_1, \dots, x_m)$ where $\varphi$ is a quantifier-free formula consisting of conjunctions of equalities and inequalities of real polynomials. To prove membership in $\exists \mathbb{R}$, use $x_1, \dots, x_m$ to represent the coordinates of the pieces' vertices in $P$ and $Q$. Then, use $\varphi$ to verify that the pieces are well-defined partitions of $P$ and $Q$ and that each piece in $P$ is a transformation of a piece in~$Q$; these conditions can be written as polynomial (in)equalities of degree at most two. $\exists \mathbb{R}$ is known to be in PSPACE~\cite{Canny-1988}.
%Hence, under this restriction, \textsc{$k$-Piece Dissection} is decidable.
	\fi
	\item Is \textsc{$k$-Piece Dissection} still hard when one or both of the input polygons are required to be convex?
	\item Can we prove stronger hardness of approximation, or find an approximation algorithm, for \textsc{Min Piece Dissection}? The current best known algorithm for finding a dissection is a worst-case bound of a pseudopolynomial number of pieces~\cite{FeatureSize_EGC2011f}.
%As such, it is not even clear whether {\em any} constant approximation ratio is achievable in polynomial time. Thus, one would expect to be able to prove a stronger hardness of approximation result than ours for \textsc{Min Piece Dissection}.
        \item Is \textsc{$k$-Piece Dissection} solvable in polynomial time for constant~$k$?  Membership in FPT would be ideal, but even XP would be interesting.
\end{itemize}

\section*{Acknowledgments}

We thank Greg Frederickson for helpful discussions.

\bibliographystyle{splncs03}
\bibliography{kpiece}

\appendix
\magicappendix

\end{document}